\documentclass[runningheads,envcountsame]{llncs}
\usepackage{graphicx}
\usepackage{color}

\usepackage[draft,mode=multiuser]{fixme}
\usepackage[hidelinks]{hyperref}

\newcommand{\ms}{\mathit{MS}}
\newcommand{\propSq}{\mathit{P}_{\mathrm{sqr}}}
\newcommand{\propPal}{\mathit{P}_{\mathrm{pal}}}
\newcommand{\propRn}{\mathit{P}_{\mathrm{per}}}
\newcommand{\propLyn}{\mathit{P}_{\mathrm{Lyn}}}
\newcommand{\propNSq}{\mathit{P}_{\mathrm{sqf}}}
\newcommand{\ST}{\mathit{ST}}
\newcommand{\pathstr}{\mathit{str}}

\begin{document}
\title{On Longest Common Property Preserved Substring Queries}
\author{Kazuki Kai\inst{1} \and
Yuto Nakashima\inst{1}\and
Shunsuke Inenaga\inst{1}\and
Hideo~Bannai\inst{1}\orcidID{0000-0002-6856-5185}\and
Masayuki~Takeda\inst{1}
\and Tomasz~Kociumaka\inst{2,3}\orcidID{0000-0002-2477-1702}
}
\authorrunning{K. Kai et al.}
\institute{
Department of Informatics, Kyushu University, Japan\\
\email{kai.kazuki.640@s.kyushu-u.ac.jp}\\
\email{\{yuto.nakashima,inenaga,bannnai,takeda\}@inf.kyushu-u.ac.jp}
\and
Department of Computer Science, Bar-Ilan University, Israel\\
\and
Institute of Informatics, University of Warsaw, Poland\\
\email{kociumaka@mimuw.edu.pl}
}
\maketitle              \begin{abstract}
We revisit the problem of longest common property preserving substring queries introduced by~Ayad et al. (SPIRE 2018, arXiv 2018).
We consider a generalized and unified on-line setting,
where we are given a set $X$ of $k$ strings of total length $n$ that
can be pre-processed so that, given a query string
$y$ and a positive integer $k'\leq k$, we can determine the longest substring of $y$ that satisfies some specific property
and is common to at least $k'$ strings in $X$.
Ayad et al. considered the longest square-free substring in an on-line setting and the longest periodic and palindromic substring in an off-line setting.
In this paper, we give efficient solutions in the on-line setting
for finding the longest common square, periodic, palindromic,
and Lyndon substrings.
More precisely, we show that $X$ can be pre-processed in $O(n)$ time
resulting in a data structure of $O(n)$ size that answers
queries in $O(|y|\log\sigma)$ time and $O(1)$ working space,
where $\sigma$ is the size of the alphabet, and
the common substring must be a 
square, a periodic substring, a palindrome, or a Lyndon word.

\keywords{squares \and periodic substrings \and palindromes \and Lyndon words}
\end{abstract}
\section{Introduction}
The longest common substring of two strings $x$ and $y$ is a longest string that is a substring of both $x$ and $y$.
It is well known that the problem can be solved in linear time,
using the generalized suffix tree of $x$ and $y$~\cite{DBLP:conf/focs/Weiner73,gusfield1997algorithms}.

Ayad et al.~\cite{DBLP:conf/spire/AyadBGIPPR18,DBLP:journals/corr/abs-1810-02099} proposed a class of problems called {\em longest common property preserved substring}, where
the aim is to find the longest substring that has some property and is common to a subset of the input strings.
They considered several problems in two different settings.

In the first {\em on-line} setting,
one is given a string $x$ that can be pre-processed,
and the problem is to answer, for any given query string $y$,
the longest square-free substring that is common to both $x$ and $y$.
Their solution takes $O(|x|)$ time for preprocessing and $O(|y|\log\sigma)$ 
time for queries, where $\sigma$ is the alphabet size.

In the second {\em off-line} setting, one is given a set of $k$ strings of total length $n$ and a positive integer $k' \leq k$, and the problem is to find the longest periodic substring, as well as the longest palindromic substring,
that is common to at least $k'$ of the strings.
Their solution works in $O(n)$ time and space.
However, it does not (at least directly) give a solution 
for the on-line setting.

In this paper, we consider a generalized and unified on-line setting, where we are given a set $X$ of
$k$ strings with total length $n$ that can be pre-processed, and the problem
is to answer, for any given query string $y$ and positive integer $k' \leq k$, the longest property preserved
substring that is common to $y$ and at least $k'$ of the strings.
We give solutions to the following properties in this setting, 
all working in $O(n)$ time and space preprocessing,
and $O(|y|\log\sigma)$ time and $O(1)$ working space for answering queries: squares, periodic substrings, palindromes, and Lyndon words.
Furthermore, we note that solutions for the off-line setting 
can be obtained by using our solutions for the on-line setting.
We also note that our algorithms can be modified to remove the $\log\sigma$ factor in the off-line setting.

As related work, the off-line version of property preserved subsequences 
have been considered for some properties.
The longest common square subsequence between two strings
can be computed in $O(n^6)$ time~\cite{DBLP:conf/cpm/InoueIHBT18}.
The longest common palindromic subsequence between two strings
of length $n$ can be computed in $O(n^4)$ time~\cite{BAE201829}.
 \section{Preliminaries}
\subsection{Strings}
Let $\Sigma$ be a set of symbols, or alphabet, 
and $\Sigma^*$ the set of strings over $\Sigma$.
We assume a constant or linearly-sortable integer alphabet\footnote{Note that a string
on a general ordered alphabet can be transformed into a string on an integer alphabet in $O(n\log\sigma)$ time.}
and use $\sigma$ to denote the size of the alphabet,
i.e. $|\Sigma| = \sigma$.
For any string $x\in\Sigma^*$, let $|x|$ denote its length.
For any integer $1 \leq i \leq |x|$, $x[i]$ is the $i$th symbol of $x$,
and for any integers $1 \leq i\leq j \leq |x|$, $x[i..j] = x[i]\cdots x[j]$.
For convenience, $x[i..j]$ is the empty string when $i > j$.
If a string $w$ satisfies $w = xyz$, then, strings $x$,$y$,$z$ are respectively called
a prefix, substring, and suffix of $w$.
A prefix (resp. substring, suffix) is called a {\em proper} prefix (resp. substring, suffix)
if it is shorter than the string.

Let $x^R$ denote the reverse of $x$, i.e., $x^R = x[|x|]\cdots x[1]$.
A string $x$ is said to be a {\em palindrome} if $x = x^R$.
A string $y$ is a {\em square} if $y = xx$ for some string $x$,
called the {\em root} of $y$. A string $y$ is {\em primitive} if there does not exist any $x$ such that
$y = x^k$ for some integer $k\geq 2$.
A square is called a {\em primitively rooted} square,
if its root is primitive.
An integer $p$, $1\le p \le |x|$, is called a {\em period} of a string $x$, if
$x[i] = x[i+p]$ for all $1 \leq i \leq |x|-p$.
A string $x$ is {\em periodic}, if the smallest period $p$ of $x$ is at most $|x|/2$.
A {\em run} in a string is a maximal periodic substring, i.e.,
a periodic substring $x[i..j]$ with smallest period $p$ is a run,
if the period of any string $x[i'..j']$ with $i'\leq i\leq j \leq j'$ and either $i'\neq i$ or $j'\neq j$, is not $p$.
The following is the well known (weak) periodicity lemma concerning periods:
\begin{lemma}[(Weak) Periodicity Lemma~\cite{10.2307/2034009}]\label{lem:weakperiodicity}
If $p$ and $q$ are two periods of a string $w$, and $p+q \leq|w|$, then, $\gcd(p,q)$ is also a period of $w$.
\end{lemma}

Let $\prec$ denote a total order on $\Sigma$, as well as the lexicographic order on $\Sigma^*$ it induces.
That is, for any two strings $x,y$,
$x \prec y$ if and only if either
$x$ is a prefix of $y$, or
there exist strings $w,x',y' \in \Sigma^*$ such that
$x = wx'$ and $y = wy'$, and $x'[1] \prec y'[1]$.
A string is a {\em Lyndon word}~\cite{lyndon1954burnside} if it is lexicographically
smaller than any of its non-empty proper suffixes.

\subsection{Suffix Trees}
\label{subsec:st}
The suffix tree~\cite{DBLP:conf/focs/Weiner73} $\ST(x)$ of a string $x$ is a compacted trie of the set of non-empty
suffixes of $x\$$,
where $\$$ denotes a unique symbol that does not occur in $x$.
More precisely, it is 
1) a rooted tree where edges are labeled by non-empty strings,
2) the concatenation of root to leaf paths correspond to 
all and only suffixes of $x\$$,  
3) any non-leaf node has at least two children, and 
the first letter of the label on the edges to its children are distinct.

A node in $\ST(x)$ is called an {\em explicit} node,
while a position on the edges corresponding to proper prefixes 
of the edge label are called {\em implicit} nodes.
For a (possibly implicit) node $v$ in $\ST(x)$, let $\pathstr(v)$ denote the string obtained by concatenating the edge labels on the path from the root to $v$.
The {\em locus} of a string $p$ in $\ST(x)$ is a (possibly implicit) node $v$
in $\ST(x)$ such that $\pathstr(v) = p$.
Each explicit node $v$ of the suffix tree can be augmented with a 
{\em suffix link}, that points to the node $u$, such that
$\pathstr(v) = \pathstr(v)[1]\pathstr(u)$. 
It is easy to see that because $v$ is an explicit node,
$u$ is also always an explicit node.

It is well known that
the suffix tree (and suffix links) can be represented in $O(|x|)$ space and
constructed in $O(|x|\log \sigma)$ time~\cite{DBLP:journals/algorithmica/Ukkonen95} for general ordered alphabets,
or in $O(|x|)$ time for constant~\cite{DBLP:conf/focs/Weiner73} or linearly-sortable integer alphabets~\cite{DBLP:conf/focs/Farach97}.
The suffix tree can also be defined for a set of strings $X = \{ x_1, \ldots, x_k\}$, and again can be constructed in 
$O(n\log\sigma)$ time for general ordered alphabets or in $O(n)$
time for constant or linearly-sortable integer alphabets where $n$ is 
the total length of the strings.
This is done by considering and building the suffix tree for
the string $x_1\$\cdots x_k\$$ and pruning edges below any $\$$.
It is also easy to process $\ST(X)$ to compute for each explicit node $v$,
the length $|\pathstr(v)|$, as well as an occurrence $(s,b)$ of $\pathstr(v)$ in $X$,
where $x_s[b..b+|\pathstr(v)|-1] = \pathstr(v)$.
Also, these values can be computed in constant time for any implicit node,
given the values for the closest descendant explicit node.

We will later use the following lemma to efficiently add explicit nodes corresponding to the loci of a set of substrings that are of interest.
\begin{lemma}[{\cite[Corollary 7.3]{DBLP:journals/corr/abs-1107-2422}}]\label{lemma:offlineSubwordLoci}
  Given a collection of substrings $s_1, \ldots, s_k$ of a string
  $w$ of length $n$, each represented by an occurrence in $w$, in 
  $O(n + k)$ total time we can compute the locus of each substring $s_i$ 
  in the suffix tree of $w$.
  Moreover, these loci can be made explicit in $O(n + k)$ extra time.
\end{lemma}

For a string $x$ of length $n$, a {\em longest extension query}, given positions $1\leq i,j\leq n$ asks for the longest common prefix
between $x[i..n]$ and $x[j..n]$. It is known that the string can
be pre-processed in $O(n)$ time so that the longest extension query
can be answered in $O(1)$ time for any $i,j$ (e.g.~\cite{DBLP:conf/cpm/FischerH06}).

\subsection{Matching Statistics}
For two strings $x$ and $y$, the {\em matching statistics}
of $y$ with respect to $x$ is an array 
$\ms_{y,x}[1..|y|]$, where
\[\ms_{y,x}[i] = \max \{ l \geq 0 : \exists_{j\in\{1,\ldots,|x|\}}  \, x[j..j+l-1] = y[i..i+l-1]\}\]
for any $1 \leq i \leq |y|$.
That is, for each position $i$ of $y$, $\ms_{y,x}[i]$ is the length of the longest prefix of $y[i..|y|]$ that occurs in $x$.
The concept of matching statistics can be generalized to a set of strings.
For a set $X = \{ x_1, \ldots ,x_k \}$ of $k$ strings and a string $y$,
the $k'$-matching statistics of $y$ with respect to $X$ is an array
$\ms_{y,X}^{k'}[1..|y|]$ where
\[\ms_{y,X}^{k'}[i] = \max \{ l \geq 0 : 
|\{x \in X : \exists_{j\in\{1,\ldots,|x|\}} \,x[j..j+l-1] = y[i..i+l-1]\}|\ge k' \}
\]
That is, for each position $i$ of $y$, $\ms_{y,X}^{k'}[i]$ is the length of the longest prefix of $y[i..|y|]$ that occurs in at least $k'$ of the strings in $X$.

\subsection{Longest Common Property Preserved Substring Queries}
Let a function $P: \Sigma^* \rightarrow \{ \mathtt{true}, \mathtt{false}\}$ be called a 
{\em property function}.
In this paper, we will consider the following property functions
$\propNSq, \propSq, \propRn,\propPal, \propLyn$, which return $\mathtt{true}$ if and only if
a string is respectively a square-free, square, periodic string, palindrome, or a Lyndon word.

The following is the on-line version of the 
problem considered in~\cite{DBLP:conf/spire/AyadBGIPPR18}, 
where a solution was given for the longest common square free substring,
i.e., $P =\propNSq$.
\begin{problem}[Longest common property preserved substring query]
  \label{problem:single}
  Let $P$ be a property function.
  Consider a string $x$ which can be pre-processed.
  For a query string $y$, compute the longest string $z$ that is a substring of both $x$ and $y$, and also satisfies $P(z) = \mathtt{true}$.
\end{problem}
The following is the generalized version of the on-line setting that we consider in this paper.
\begin{problem}[Generalized longest common property preserved substring query]
  \label{problem:general}
  Let $P$ be a property function.
  Consider a set of strings $X = \{ x_1, \ldots, x_k\}$ that can be pre-processed.
  For a query string $y$ and positive integer $k' \leq k$,
  compute the longest substring $z$ of $y$ that is a substring of at least $k'$ strings in $X$, and also satisfies $P(z) = \mathtt{true}$.
\end{problem}

Below is a summary of our results. Here,
the {\em working space} of the query is the amount of memory that is required in excess to the data structure 
constructed in the pre-processing.
All memory other than the working space can be considered as read-only.

\begin{theorem}\label{theorem:maintheorem}
  For any property function $P\in\{ \propNSq, \propSq, \propRn, \propPal, \propLyn \}$,
  Problem~\ref{problem:general} can be answered in $O(|y|\log\sigma)$ time and $O(1)$ working space
  by constructing an $O(n)$ space data structure in $O(n)$ time,
  where $n = \sum_{i=1}^k|x_i|$ is the total length of the strings in $X$.
\end{theorem}
We further note that our algorithms
do not require random access to $y$ during the query
and thus work in the same time/space bounds
even if each symbol of $y$ is given in a streaming fashion.

The proof of the Theorem for each property function is given in the next section.
 \section{Algorithms}
In this section we present our algorithm.
The outline of our solutions is as follows.

The preprocessing consists of the following steps:
\begin{enumerate}
  \item Construct the generalized suffix tree $\ST(X)$ of $X$.\label{preprocess:st} 
  \item For each explicit node $v$ of $\ST(X)$, compute the number of strings in $X$ that contain $\pathstr(v)$ 
as a substring.\label{preprocess:css} 
  \item Process $\ST(X)$ and construct a data structure
  so that, 
  given any position on $\ST(X)$, we can efficiently find a 
  {\em candidate} for the solution.\label{preprocess:function}
\end{enumerate}
Then, queries are answered as follows:
\begin{enumerate}\setcounter{enumi}{3}
  \item For each position $i$ in $y$, compute $\ms_{y,X}^{k'}[i]$, i.e., 
  the $k'$-matching statistics of $y$ with respect to $X$,
  as a locus $v_i$ of $y[i..e_i]$ in $\ST(X)$.\label{query:kpms}
  \item For each such locus $v_i$, compute a candidate using the data 
  structure computed in Step~\ref{preprocess:function} of the pre-processing.\label{query:cand}
  \item Output the longest string computed in the previous step.\label{query:output}
\end{enumerate}

As mentioned in Section~\ref{subsec:st},
Step~\ref{preprocess:st} can be performed in $O(n)$ time and space. 
The task beyond Step~\ref{preprocess:css} is known as the color set size problem, 
and can also be executed in $O(n)$ time~\cite{DBLP:conf/cpm/Hui92}.

Using $\ST(X)$, Step~\ref{query:kpms}, i.e., each locus $v_i$ of the substring $y[i..e_i]$ 
of $y$ that corresponds to each element of the matching statistics of $y$ with respect to $X$
(i.e., $e_i-i+1 = \ms_{y,X}^{k'}[i]$)
can be computed in $O(|y|\log\sigma)$ time and $O(1)$ working space, 
with a minor modification to the algorithm for computing the matching statistics between single strings \cite[Theorem~7.8.1]{gusfield1997algorithms}.
The algorithm for a single string first traverses the suffix tree with the string $y[1..|y|]$ as long as possible to compute the locus corresponding to $\ms_{y,x}[1]$. Let this prefix be $y[1..e_1]$.
Given a locus $v_i$ of $y[i..e_i]$ for some $1 \leq i < |y|$, 
the suffix link of the closest ancestor 
of $v_i$ is used in order to first efficiently find the locus of $y[i+1..e_i]$.
Then, the suffix tree is further traversed to obtain the locus of $y[i+1..e_{i+1}]$.
The time bound for the traversal follows from a well known amortized analysis
which considers the depth of the traversed nodes, 
similar to that in the online construction of suffix trees~\cite{DBLP:journals/algorithmica/Ukkonen95}.
For computing $\ms_{y,X}^{k'}$, we can simply imagine that
subtrees below edges leading to a node $v$ of $\ST(X)$ are pruned
if $\pathstr(v)$ is not contained in at least $k'$ strings of $X$.
This can be done by aborting the traversal of $y[i..|y|]$
when we encounter such an edge, by using the information obtained
in Step~\ref{preprocess:css}.
It is easy to see that the algorithm still works in the same time bound,
since the suffix link of any remaining node still points to a node
that is not pruned (i.e., if a string is contained in at least $k'$ strings,
its suffix will also be contained in the same strings).
Thus, we can visit each locus corresponding to $\ms_{y,X}^{k'}$
in $O(|y|\log\sigma)$ total time and $O(1)$ working space.

The crux of our algorithm is therefore in the details of 
Step~\ref{preprocess:function} and Step~\ref{query:cand}:
designing what the candidates are and how to compute them given $v_i$.
Notice that the solution is the longest string $z$ which satisfies $P(z) = \mathtt{true}$ and is a substring of $y[i..e_i]$ 
for some $i = 1, \ldots, n$.

\subsection{Square-free substrings}
Ayad et al.~\cite{DBLP:conf/spire/AyadBGIPPR18,DBLP:journals/corr/abs-1810-02099}
gave a solution to the on-line longest common square-free substring query problem 
for a single string (Problem~\ref{problem:single}),
in $O(|x|)$ time and space preprocessing and $O(|y|\log\sigma)$ time and $O(1)$ space query. 
We note that their algorithm easily extends to the generalized version (Problem~\ref{problem:general}). 
The only difference lies in that $\ms_{y,X}^{k'}$ is computed
instead of $\ms_{y,x}$, which can be done in $O(|y|\log\sigma)$ time 
and $O(1)$ space, as described above.
We leave the other details to~\cite{DBLP:conf/spire/AyadBGIPPR18,DBLP:journals/corr/abs-1810-02099}.

\subsection{Squares}
\label{subsec:squares}
As mentioned in the introduction,
Ayad et al.~\cite{DBLP:conf/spire/AyadBGIPPR18,DBLP:journals/corr/abs-1810-02099} also considered 
longest common periodic substrings, but in the off-line setting.
However, their algorithm is not readily extendible to the on-line setting.
It relies on the fact that the ending position of a longest common periodic 
substring must coincide with an ending position of some run in the set of strings,
and solved the problem by computing all loci corresponding to runs in $X$.
To utilize this observation in the on-line setting,
the loci of all runs in the query string $y$
must be identified in $\ST(X)$, which seems difficult to do in time not depending on $X$.

Here, we first show how to efficiently solve the problem in the on-line setting for squares,
and then we extend that solution to obtain an efficient solution for periodic substrings.
Below is an important property of squares which we exploit.\begin{lemma}[{\cite[Theorem~1]{DBLP:journals/jct/FraenkelS98}}]\label{lemma:AtMostTwoRightmostOccurrences}
   A given position can be the right-most occurrence of at most two distinct squares.
\end{lemma}
It follows from the above lemma that the number of distinct squares in a given string is 
linear in its length~\cite{DBLP:journals/jct/FraenkelS98}.
Also, it gives us the following Corollary.
\begin{corollary}\label{corollary:squaresOnEdge}
  There can only be at most two implicit nodes on an edge of a suffix tree which corresponds to a square.
  \qed
\end{corollary}
\begin{proof}
  The right-most occurrence of a square is the maximum position corresponding to leaves in the subtree
  rooted at the square. Since implicit nodes that correspond to squares on the same edge share the right-most
  occurrence, a third implicit node would contradict Lemma~\ref{lemma:AtMostTwoRightmostOccurrences}.
\end{proof}

For squares, we first compute the locus in $\ST(X)$ of all distinct squares that are substrings of 
strings in $X$,
and we make them explicit nodes in $\ST(X)$. Note that these additional nodes will not have suffix links,
but since there are only a constant number of them on each edge of the original suffix tree,
it will only add a constant factor to the amortized analysis when computing Step~\ref{query:kpms}.
The loci of all squares can be computed in $O(n)$ total time using 
the algorithm of~\cite{DBLP:conf/cpm/BannaiIK17}.
We further add to each explicit node in $\ST(X)$ (including the ones we introduced above)
a pointer to the nearest ancestor that is the locus of a square (see Fig.~\ref{figure:st-sq} for an example of the pointers).
Notice that a node that is the locus of a square is explicit and points to itself.
This can be also done in linear time by a simple depth-first traversal on $\ST(X)$.
\begin{figure}[tbp]
  \begin{center}
    \includegraphics[width=0.9\textwidth]{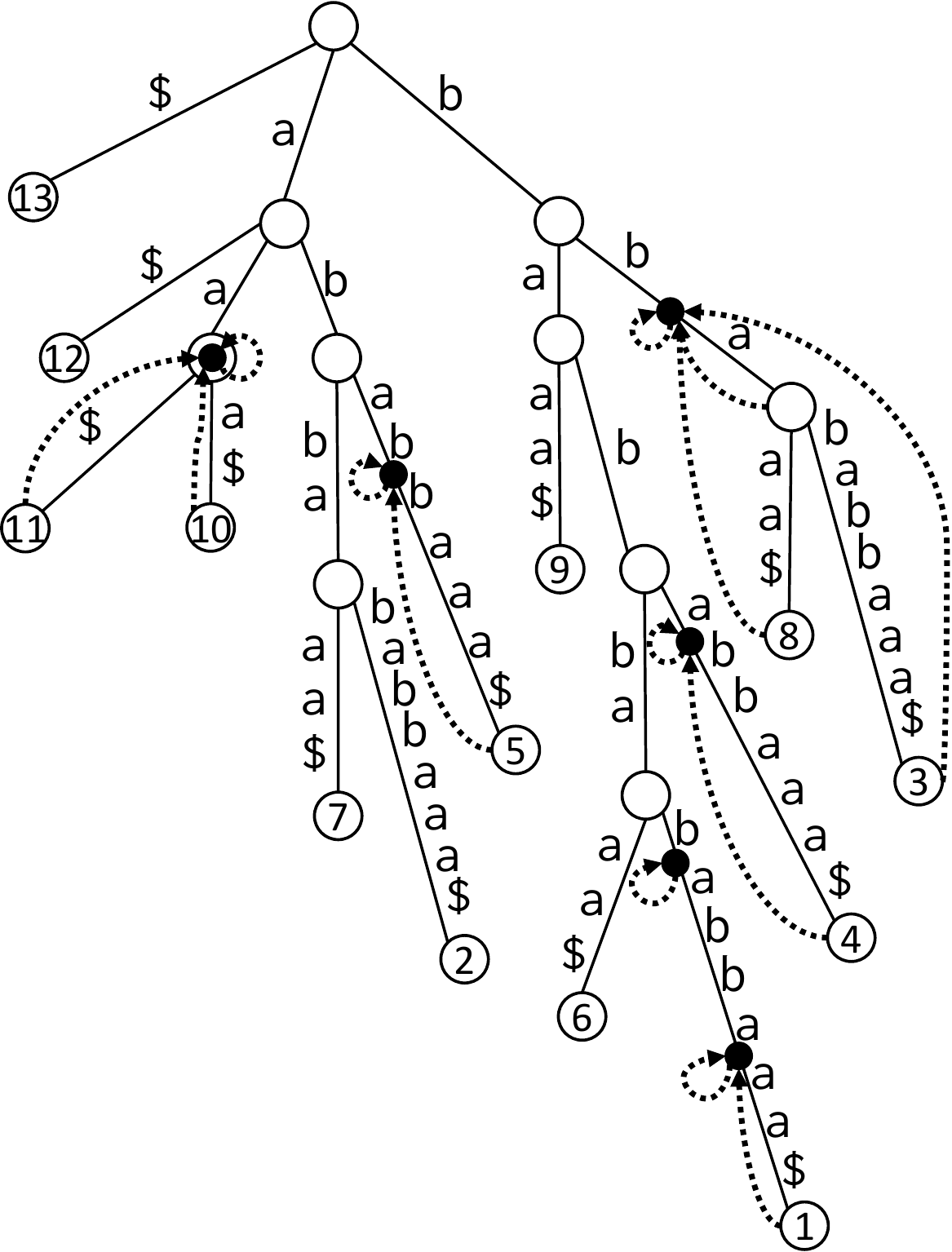}
  \end{center}
  \caption{Example of $\ST(X)$ for a single string $X = \{\mathtt{babbababbaaa\$}\}$
  augmented with nodes and pointers for $P = \propSq$.
  The solid dark circles show the loci corresponding to squares which are made explicit.
  The dotted arrows show the pointers which point to the nearest ancestor which is a square. 
  Pointers which point to the root (i.e., there is no non-empty prefix that is a square) are omitted.
  }
  \label{figure:st-sq}
\end{figure}

The candidate, which we will compute in Step~\ref{query:cand} for each locus $v_i$,
is the longest square that is a prefix of $y[i..e_i]$. 
It can be determined for each locus $v_i$ by using the pointers.
When $v_i$ is an explicit node, the pointer of $v_i$ is the answer.
When $v_i$ is an implicit node, the pointer of the parent of $v_i$ is the answer.
The longest such square for all loci is the answer to the query.
This is because the longest common square must be a prefix of the string corresponding to the $k'$-matching statistics of some position.
Thus, we have a solution as claimed in Theorem~\ref{theorem:maintheorem} for $P = \propSq$.

\subsection{Periodic Substrings}
Next, we extend the solution for squares to periodic substrings as follows.

We first explain the data structure,
which is again an augmented $\ST(X)$.
For each primitively rooted square substring $w$, we make the locus of $w$ in $\ST(X)$ an explicit node.
(The non-primitively rooted squares are redundant since they will lead to the 
same periodic substrings.)
Furthermore, we also make explicit the deepest locus $v$ in $\ST(X)$ obtained
by periodically extending a primitively rooted square $w$, i.e., 
$w$ is a prefix of $\pathstr(v)$ and the smallest period of $\pathstr(v)$ is $\frac12|w|$.
We add to each explicit node, a pointer to the nearest explicit ancestor (including itself)
that is a locus of some square or its extension.
If an explicit node is an extension of a square, it will also hold information to identify 
which square it is an extension of (e.g., its period).
Note that the pointer of an explicit node that lies between a square 
and its extension will point to itself.
Fig.~\ref{fig:stsqexample} shows an example of $\ST(X)$
for a single string $X = \{ \mathtt{aababababbababab\$}\}$,
where loci corresponding to squares and their rightmost-maximal extensions are depicted.

\begin{figure}[tbp]
  \begin{center}
    \includegraphics[width=0.9\textwidth]{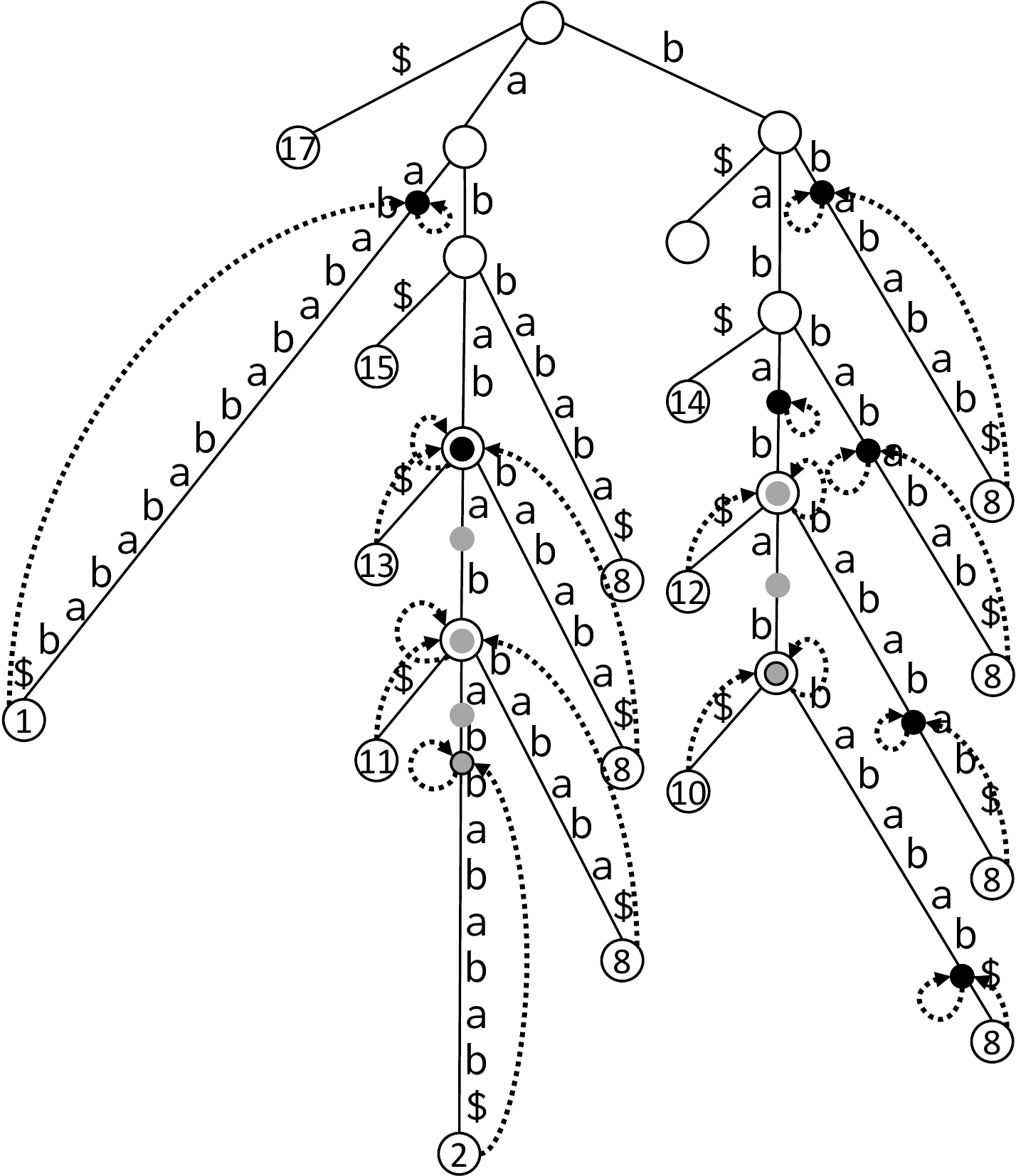}
  \end{center}
  \caption{Example of $\ST(X)$ for a single string $X = \{ \mathtt{aababababbababab\$} \}$
  augmented with nodes and pointers for $P=\propRn$.
  The solid dark circles show the loci corresponding to squares, which are made explicit,
  and the grey circles show the loci corresponding to their extensions, 
  where the ones with a solid border are made explicit.
  The dotted arrows show the pointers which point to the nearest ancestor (longest prefix)
  which is periodic. The pointers which point to the root are omitted.
  The total number of solid dark circles, as well as grey circles with a solid border is $O(n)$.
  The total number of implicit grey circles without the solid border is not necessarily $O(n)$, but
  since they occur consecutively from a solid node, 
  they can be represented in $O(n)$ space.
  }
  \label{fig:stsqexample}
\end{figure}

We first show how the above augmentation of $\ST(X)$
can be executed in $O(n)$ time.
We first make explicit all loci of squares as in Section~\ref{subsec:squares}, which can be done in $O(n)$ time.
Then, we start from the locus of each primitively rooted square
and extend the periodicity of the square towards the leaves 
of $\ST(X)$ as deep as possible.
For each explicit node we encounter during this
extension, the pointer will point to itself, and the node will also
store the period of the underlying square.
The total cost of this extension can be bounded as follows.
Due to the periodicity lemma (Lemma~\ref{lem:weakperiodicity}), 
any locus of a primitively rooted square or its extension cannot 
be an extension of a shorter square; if it were,
a (proper) divisor of the period of the longer square would 
also be a period of the square,
and would contradict that the longer square is primitively rooted.
Thus, we can naturally discard all non-primitive squares
by doing the extensions in increasing order of the length 
of squares (which can be obtained in linear time by radix sort).
During the extension, if the square is non-primitive, then 
the pointers of nodes will already have been determined by a square
of a shorter period, and this situation can be detected.
From the above arguments, any edge of $\ST(X)$ is traversed at most once, 
and thus the extension can be done by traversing a total of $O(n)$ 
nodes and edges.
Because we know the period of the square we wish to extend,
and for any locus $v$ in $\ST(X)$, an occurrence $b$ in some $x_s\in X$,
we can compute the extension in $O(1)$ time per edge of $\ST(X)$
by using longest common extension queries.
Therefore, the total time for building the augmented $\ST(X)$ can be done in $O(n)$ time.

Queries can be answered in the same way as for squares,
except for a small modification.
For any locus $v_i$, if $v_i$ is an explicit node, then the pointer
of $v_i$ gives the answer.
When $v_i$ is an implicit node, we check $v_i$'s parent and its immediate descendant. 
If these nodes are both extensions originating from the same square (i.e., their labels have the same period), 
then, the answer is $v_i$ itself, since it is also an extension of the square. Otherwise, the pointer of the parent node provides the answer.

Thus, we have a solution as claimed in Theorem~\ref{theorem:maintheorem}
for $P = \propRn$.

\subsection{Palindromes}
It is well known that the number of non-empty distinct palindromes in a string of length $n$ is at most $n$,
since only the longest palindrome starting at some position can be the right-most occurrence of that palindrome
in the string.
\begin{lemma}[{\cite[Proposition~1]{Droubay:2001ct}}]
  A given position can be the right-most occurrence of at most one distinct palindrome.
\end{lemma}
All distinct palindromes in a string can be computed in linear time~\cite{GROULT2010908}.
The locus of each palindrome in $\ST(X)$ can be computed in $O(n)$ time by Lemma~\ref{lemma:offlineSubwordLoci}.
The rest is the same as for squares;
we make all loci corresponding to a palindrome an explicit node,
and do a linear time depth-first traversal on $\ST(X)$ to make pointers
to the nearest ancestor that is a palindrome.
As in the case of squares, we can bound the number of palindromes which will be on an edge of 
the original suffix tree.
\begin{corollary}\label{corollary:palindromesOnEdge}
  There can only be at most one implicit node on an edge of a suffix tree that corresponds to a palindrome.
\end{corollary}
\begin{proof}
 Analogous to the proof of Corollary~\ref{corollary:squaresOnEdge}. 
 \qed
\end{proof}

The rest of the algorithm and analysis is the same,
thus we obtain a solution for $P=\propPal$ of Theorem~\ref{theorem:maintheorem}.

\subsection{Lyndon Words}
For Lyndon words,
we use the following result by Kociumaka~\cite{DBLP:conf/cpm/Kociumaka16}.
A {\em minimal suffix query} (MSQ) on a string $T$,
given indices $\ell,r$ such that $1\le \ell \le r \le |T|$,
determines the lexicographically smallest suffix of $T[\ell..r]$.
\begin{lemma}[{\cite[Theorem 17]{DBLP:conf/cpm/Kociumaka16}}]\label{lemma:minSuffixQuery}
  For any string $T$ of length $n$, there exists a data structure of size $O(n)$ which can be constructed in $O(n)$ time, which can answer minimal suffix queries on $T$ in constant time.
\end{lemma}
Using Lemma~\ref{lemma:minSuffixQuery}, we can find the longest Lyndon word ending at a given position.\begin{lemma}
  For any string $w$, the lexicographically smallest
  suffix is the longest Lyndon word that is a suffix of $w$.
\end{lemma}
\begin{proof}
  From the definition of Lyndon words, 
  it is clear that the minimal suffix must be a Lyndon word,
  and that a longer suffix cannot be a Lyndon word. \qed
\end{proof}

For Step~\ref{preprocess:function}, we process all strings in $X$ so that MSQ can be answered in constant time.

For Step~\ref{query:cand}, the situation is a bit different from
squares, periodic strings, and palindromes, in that we 
can have multiple candidates for each $v_i$ rather than just one.
For convenience, let $e_0 = 0$.
For each $0 \leq i < n$, suppose we have obtained the locus $v_i$ of $y[i..e_i]$,
and the next locus $v_{i+1}$ of $y[i+1..e_{i+1}]$.
Notice that $e_i \leq e_{i+1}$ and for all positions $e'$ such that $e_i < e' \leq e_{i+1}$,
$y[i+1..e']$ is the longest substring of $y$ that ends at $e'$ and is a substring of $y[j..e_j]$ for some $j=1,\ldots,n$.
As mentioned in Section~\ref{subsec:st},
we can obtain an occurrence $(s,b)$ of $y[i+1..e_{i+1}]$ such that $y[i+1..e_{i+1}] = x_s[b..b+\ms_{y,X}^{k'}[i]-1]$.
Then, we use MSQ on substrings $x_s[b..r]$ such that $b + e_i - i \leq r \leq b + e_{i+1} - (i+1)$,
which is equivalent to using MSQ on substrings $y[i+1..e']$ for all $e_i < e' \leq e_{i+1}$.
The longest suffix Lyndon word obtained in all the queries is therefore
the longest Lyndon word that is a substring of $y[j..e_j]$ for some $j = 1,\ldots, n$.
Since we perform $e_{i+1} - e_i$ MSQs for each position $i+1$ of $y$, the total number of MSQs is $|y|$,
which takes $O(|y|)$ time.
Thus, we have a solution as claimed in Theorem~\ref{theorem:maintheorem} for $P = \propLyn$.

\subsection{Solutions in the Off-line Setting}
We note that a solution for the on-line setting gives a
solution for the off-line setting, since,
for any $X = \{ x_1, \ldots, x_k \}$,
we can consider the string $y = x_1\#\cdots\#x_k$, 
where $\#$ is again a symbol that doesn't appear elsewhere.
Since $|y| = O(n)$, the preprocessing time is $O(n)$,
and the query time is $O(n\log\sigma)$.

Furthermore, we can remove the $\log\sigma$ factor
by processing $\ST(X)$ for {\em level ancestor queries}.
Level ancestor queries, given a node $v$ of tree $T$ and integer $d$,
answer the ancestor of $v$ at (node) depth $d$.
It is known that level ancestor queries can be answered in constant time, after
linear time preporocessing of $T$~(e.g.~\cite{DBLP:journals/tcs/BenderF04}).
The $\log\sigma$ factor came from determining which child of 
a branching node we needed to follow when traversing $\ST(X)$
with some suffix of $y$.
Since, in this case, we can identify the leaf in $\ST(X)$ 
that corresponds to the current suffix of $y$ 
(i.e. some suffix of $x_i$ in $X$) that is being traversed,
we can use level ancestor query to determine, in constant time,
the child of the branching node that is an ancestor of that leaf,
thus getting rid of the $\log\sigma$ factor. 
\section{Conclusion}
We considered the generalized on-line variant of the longest common property preserved substring 
problem proposed by Ayad et al.~\cite{DBLP:conf/spire/AyadBGIPPR18,DBLP:journals/corr/abs-1107-2422},
and 1) unified the two problem settings, and 
2) proposed algorithms for several properties, namely,
squares, periodic substrings, palindromes, and Lyndon words.
For all these properties, we can answer queries in
$O(|y|\log\sigma)$ time and $O(1)$ working space,
with $O(n)$ time and space preprocessing.

 \section*{Acknowledgments}
This work was supported by JSPS KAKENHI Grant Numbers JP18K18002 (YN), JP17H01697 (SI), JP16H02783 (HB), and JP18H04098 (MT).
Tomasz Kociumaka was supported by ISF grants no. 824/17 and 1278/16 and by an ERC grant MPM under the EU's Horizon 2020 Research and Innovation Programme (grant no. 683064). \bibliographystyle{splncs04}
\bibliography{refs}
\end{document}